\documentclass[journal]{IEEEtran}
\newcounter{mytempeqncnt}
\newtheorem{thm}{Theorem}
\newtheorem{lem}{Lemma}

\newenvironment{proof}[1][Proof]{\begin{trivlist}
\item[\hskip \labelsep {\bfseries #1}]}{\end{trivlist}}

\newcommand{\qed}{\nobreak \ifvmode \relax \else
      \ifdim\lastskip<1.5em \hskip-\lastskip
      \hskip1.5em plus0em minus0.5em \fi \nobreak
      \vrule height0.75em width0.5em depth0.25em\fi}

\ifCLASSINFOpdf
  \usepackage[pdftex]{graphicx}
  \graphicspath{{./}{Figs/}}
  \DeclareGraphicsExtensions{.pdf,.jpeg,.png}
\else
\fi
\usepackage{amssymb}
\usepackage{amsmath}
\usepackage{mathtools}
\usepackage{bigdelim}

\newcommand{\MeijerG}[7]{G^{#1,#2}_{#3,#4} \left( \begin{smallmatrix} #5 \\ #6 \end{smallmatrix} \middle\vert #7 \right) }

\begin{document}
\title{Opportunistic Beamforming using Dumb Basis Patterns in Multiple Access Cognitive Channels}

\author{Ahmed~M.~Alaa,~\IEEEmembership{Student Member,~IEEE}, Mahmoud~H.~Ismail,~\IEEEmembership{Member,~IEEE}, and Hazim Tawfik 
\thanks{The authors are with the Department
of Electronics and Electrical Communications Engineering, Cairo University, Gizah,
12316, Egypt (e-mail: \{aalaa\}@eece.cu.edu.eg).}
\thanks{Manuscript received XXXX XX, 2013; revised XXXX XX, 201X.}}

\markboth{IEEE Transactions on XXXX,~Vol.~XX, No.~X, XXXX~201X}%
{Alaa \MakeLowercase{\textit{et al.}}:Opportunistic Beamforming using Dumb Basis Patterns in Multi-Access Cognitive Channels}

\maketitle
\begin{abstract}
In this paper, we investigate multiuser diversity in interference-limited Multiple Access (MAC) underlay cognitive channels with Line-of-Sight interference (LoS) from the secondary to the primary network. It is shown that for $N$ secondary users, and assuming Rician interference channels, the secondary sum capacity scales like $\log\left(\frac{K^{2}+K}{\mathcal{W}\left(\frac{K e^{K}}{N}\right)}\right)$, where $K$ is the $K$-factor of the Rician channels, and $\mathcal{W}(.)$ is the Lambert W function. Thus, LoS interference hinders the achievable multiuser diversity gain experienced in Rayleigh channels, where the sum capacity grows like $\log(N)$. To overcome this problem, we propose the usage of single radio Electronically Steerable Parasitic Array Radiator (ESPAR) antennas at the secondary mobile terminals. Using ESPAR antennas, we induce artificial fluctuations in the interference channels to restore the $\log(N)$ growth rate by assigning random weights to orthogonal {\it basis patterns}. We term this technique as {\it Random Aerial Beamforming} (RAB). While LoS interference is originally a source of capacity hindrance, we show that using RAB, it can actually be exploited to improve multiuser interference diversity by boosting the {\it effective number of users} with minimal hardware complexity.
\end{abstract}


\IEEEpeerreviewmaketitle
\section{Introduction}
\IEEEPARstart{S}{ignificant} interest has recently been devoted to the capacity analysis of underlay cognitive radio systems in fading environments. In underlay cognitive radio, a Secondary User (SU) aggressively transmits its data over the Primary User (PU) channel while keeping the interference experienced by the PU below a predefined {\it interference temperature} \cite{1}. Moreover, in multiuser cognitive networks, Dynamic Time Division Multiple Access (D-TDMA) can be used such that the SU with the best Signal-to-Interference Ratio (SINR) is scheduled to transmit at each time slot. Multiuser cognitive networks were thoroughly studied in \cite{2}-\cite{5}. While the sum capacity of conventional non-cognitive multiuser networks scales with $\log(\log(N))$ \cite{6}, it was shown in \cite{3} that the SUs capacity of {\it interference-limited} underlay cognitive networks scales with $\log(N)$. This improvement in the capacity growth rate is due to {\it Multiuser Interference Diversity} (MID) introduced in \cite{2}, which results from the opportunities offered by the fluctuations in the interference channels. Because a Line-of-Sight (LoS) channel has a poor dynamic range for the channel gain fluctuations, it is usually considered as a source of multiuser diversity hindrance. In the seminal work of Viswanath {\it et al.} \cite{6}, it was shown that the sum capacity for an $N$-user downlink cellular network grows like $\log\left(\left(\sqrt{\frac{1}{K+1}\log(N)}+\sqrt{\frac{K}{K+1}}\right)^{2}\right)$ when the channels between the base station and the mobile users are Rician with a $K$-factor of $K$. For such network, dumb antennas were used to induce artificial channel fluctuations in order to improve the capacity scaling characteristics. Given that LoS channels hinder multiuser diversity in non-cognitive networks, it is expected that LoS interference would also hinder MID in underlay cognitive channels. In this paper, we study the capacity scaling characteristics for an underlay Multiple Access (MAC) cognitive channel with Rician LoS interference. It is shown that the SUs sum capacity scales like $\log\left(\frac{K^{2}+K}{\mathcal{W}\left(\frac{K e^{K}}{N}\right)}\right)$, where $K$ is the Rician $K$-factor. For moderate values of $K$, this growth rate can be approximated as $\log\left(\frac{N(K+1)}{e^{K}}\right)$, which is equivalent to the growth rate of an underlay MAC with Rayleigh faded interference channels and $\frac{N(K+1)}{e^{K}}$ effective SUs. For large values of $K$, the scaling law tends to $\log(\log(N))$, which corresponds to the case when the interference channels are deterministic and no MID is exploited.

In order to overcome the negative impact of LoS interference on capacity scaling, and following the same line of thought in \cite{6}, we propose the usage of {\it Random Beamforming} to induce artificial fluctuations in the Rician interference channels. Because the primary users are considered to be oblivious to the secondary network, and deploying multiple dumb antennas at the base station will not affect the interference channels fluctuations, it is essential to apply random beamforming using dumb antennas at the mobile terminals. While deploying dumb antennas at a base station is bearable, it can not be tolerated in a low-cost mobile terminal with tight space limitations. Therefore, instead of manipulating Rician channels by assigning random weights to multiple dumb antennas, we adopt a random beamforming technique that assigns random weights to multiple {\it dumb basis patterns} of an Electronically Steerable Parasitic Array Radiator (ESPAR) antenna. An ESPAR antenna involves a single RF chain and has a reconfigurable radiation pattern that is controlled by assigning arbitrary weights to $M$ orthonormal basis radiation patterns via altering a set of reactive loads. Because ESPAR antennas entail a single RF chain and a compact parasitic antenna array, it is well suited for low cost mobile terminals. The Degrees-of-Freedom offered by the ESPAR basis patterns are termed as {\it Aerial DoF} \cite{7}, thus we term random beamforming using the ESPAR basis patterns as {\it Random Aerial Beamforming} (RAB). We show that using RAB, the SUs sum capacity scales with $\log(N)$ in Rician interference channels. A fundamental result of this paper is that, while LoS interference originally acts as a source of multiuser capacity hindrance, it can actually be exploited to improve MID. Using only two dumb basis patterns, i.e., one parasitic antenna element, we show that the LoS interference component can be exploited to improve the multiuser diversity compared to Rayleigh-faded interference. This improvement can be interpreted as a boost in the effective number of SUs, where we show that the capacity grows like $\log \left(\sqrt{\frac{(K+1)^{2}}{2K \pi}} N\right)$, which is equivalent to the capacity scaling of a Rayleigh-faded interference channel but with $\sqrt{\frac{(K+1)^{2}}{2K \pi}} N$ users instead of $N$. This effective number of users increases with the increase of $K$, thus LoS interference acts as a friend and not a foe in this case.

The application of the ESPAR antenna to wireless communications systems is not new. In \cite{7}-\cite{10}, the {\it Aerial DoF} provided by the orthonormal basis patterns are used to construct single Radio {\it Beamspace-MIMO} systems that can apply spatial multiplexing without the need to deploy multiple antennas. Blind interference alignment was implemented utilizing the reconfigurability feature of the ESPAR antenna in \cite{11}. Other applications of ESPAR antennas in multiuser systems can be found in \cite{12}.

The rest of the paper is organized as follows. In Section II, we present the system model. Capacity scaling for the MAC cognitive channel with LoS interference is studied in Section III. In Section IV, we present the RAB technique and study its impact on capacity scaling. Numerical results are presented in Section V and finally, the conclusions are drawn in Section VI.

\section{System Model and Notations}
\begin{figure}[!t]
\centering
\includegraphics[width=2.5in]{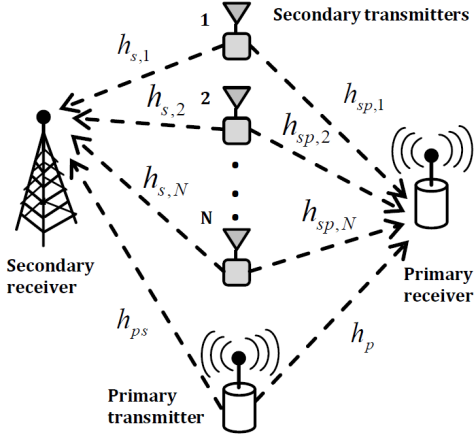}
\caption{Depiction for the multipleaccess cognitive channel.}
\label{fig_sim}
\end{figure}
This section is divided into two subsections. First, we present the system model for the cognitive MAC channel with LoS interference. Next, we explain the signal model for the ESPAR antenna.

\subsection{The cognitive MAC channel with LoS interference}

Assume an underlay MAC cognitive channel where $N$ SU mobile terminals coexist with a single PU transmit-receive pair as shown in Fig. 1. The secondary channels between the mobile users and the base station are denoted by ${\bf h_{s}} = \left(h_{s,1},h_{s,2},...,h_{s,N}\right)$, where the channel between the base station and the $i^{th}$ user is $h_{s,i} \sim \mathcal{CN}(0,\overline{\gamma}_{s})$, and all channels are independent and identically distributed. The secondary to primary network interference is characterized by a LoS Rician channel. Thus, the set of SU-to-PU channels ${\bf h_{sp}} = \left(h_{sp,1},h_{sp,2},...,h_{sp,N}\right)$ follow a Rician distribution such that $h_{sp,i} \sim \mathcal{CN}\left(\sqrt{\frac{K}{K+1}},\frac{\overline{\gamma}_{sp}}{K+1}\right)$. The PU-to-SU channel is $h_{ps}$, and has no impact on the capacity scaling characteristics, so it can be assumed to follow any distribution. The received signal at the SU receiver at the $k^{th}$ time slot is given by
\begin{equation}
\label{1}
\,\,\, r(k) = h_{s,n^{*}}(k) \sqrt{P_{s}({\bf h_{sp,n^{*}}})} x(k) + h_{ps}(k) x_{p}(k) + n(k),
\end{equation}
where $n^{*}$ is the index of the SU selected for transmission at the $k^{th}$ time slot, $r(k)$ is the signal received at the secondary base station, $x(k)$ is the symbol transmitted from the SU mobile terminal to the base station and has a unit energy, $P_{s}({\bf h_{sp,n^{*}}})$ is the SU transmit power, $x_{p}(k)$ is the PU signal and has an average energy of $\overline{\gamma}_{p}$, and $n(k) \sim \mathcal{CN}(0,1)$ is the noise signal at the base station receiver. The SU transmit power is adjusted according to a peak interference constraint $Q_{p}$ at the PU, i.e., $P_{s}({\bf h_{sp,n^{*}}}) |h_{sp,n^{*}}|^{2} \leq Q_{p}$. Letting $\gamma_{i} = |h_{i}|^{2}$ and the SU transmit power allocation $P_{s}({\bf h_{sp,n^{*}}}) = \frac{Q_{p}}{\gamma_{sp,n^{*}}}$, the SINR at the base station receiver is given by $\frac{\gamma_{s,n^{*}} \frac{Q_{p}}{\gamma_{sp,n^{*}}}}{1+\overline{\gamma}_{p} \gamma_{ps}}$. According to the D-TDMA policy, the selected user index is given by
\[n^{*} = \max_{n}\frac{\gamma_{s,n} \frac{Q_{p}}{\gamma_{sp,n}}}{1+\overline{\gamma}_{p} \gamma_{ps}}.\]
The signal model above describes the case when all SUs use conventional single antennas. In the next subsection, we study the case when ESPAR antennas are employed.

\subsection{The ESPAR signal model}

\begin{figure}[t]
\centering
\includegraphics[width=2.5in]{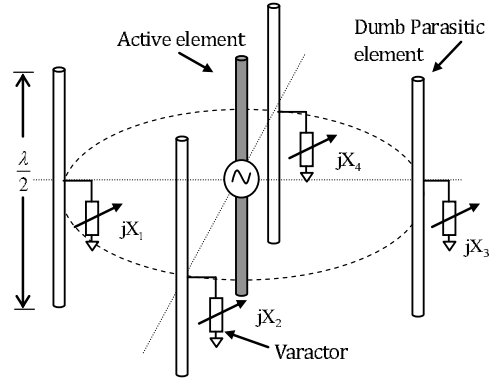}
\caption{The ESPAR antenna with 4 parasitic elements.}
\label{fig_sim}
\end{figure}

As shown in Fig. 2, an ESPAR with $M$ elements is composed of a single active element (e.g., a $\frac{\lambda}{2}$ dipole) that is surrounded by $M-1$ identical parasitic elements. Unlike multi-antenna systems, the parasitic elements are placed relatively close to the active elements. Hence, mutual coupling between different elements takes place and current is induced in all parasitic elements. The radiation pattern of the ESPAR is altered by tuning a set of $M-1$ reactive loads (varactors) $\mathbf{x} = \left[jX_{1} \ldots jX_{M-1}\right]$ attached to the parasitic elements \cite{7}. The currents in the parasitic and active elements are represented by an $M \times 1$ vector $\mathbf{i} = v_{s} (\mathbf{Y}^{-1}+\mathbf{X})^{-1}\mathbf{u}$, where ${\bf Y}$ is the $M \times M$ admittance matrix with $y_{ij}$ being the mutual admittance between the $i^{th}$ and $j^{th}$ elements. The load matrix $\mathbf{X}$ = {\bf diag}$\left(50,\,\, \mathbf{x}\right)$\footnote{The opertaion {\bf X = diag(x)} embeds a vector {\bf x} in the diagonal matrix {\bf X}.} controls the ESPAR beamforming, $\mathbf{u} = \left[1 \,\, 0 \ldots 0\right]^{T}$ is an $M \times 1$ vector and $v_{s}$ is the complex feeding at the active element \cite{8}. The radiation pattern of the ESPAR at an angle $\theta$ is thus given by $P(\theta) = \mathbf{i}^{T}\mathbf{a}(\theta)$, where $\mathbf{a}(\theta) = \left[a_{0}(\theta) \ldots a_{M-1}(\theta)\right]^{T}$ is the steering vector of the ESPAR at an angle $\theta$. The beamspace domain is a signal space where any radiation pattern can be represented as a point in this space. To represent the radiation pattern $P(\theta)$ in the beamspace domain, the steering vector $\mathbf{a}(\theta)$ is decomposed into a linear combination of a set of orthonormal basis patterns $\{\Phi_{i}(\theta)\}_{i=0}^{N-1}$ using Gram-Schmidt orthonormalization, where $N \leq M$ \cite{8}. It can be shown that the orthonormal basis patterns of the ESPAR (also known as the ADoF \cite{7}-\cite{12}) are equal to the number of parasitic elements (i.e., $N = M$). Therefore, the ESPAR radiation pattern in terms of the orthonormal basis patterns can be written as \cite{8}
\begin{equation}
\label{2}
P(\theta) = \sum_{l=1}^{M} w_{l} \Phi_{l}(\theta),
\end{equation}
where $w_{l} = \mathbf{i}^{T}\mathbf{q}_{l}$ are the weights assigned to the basis patterns and $\mathbf{q}_{l}$ is an $M \times 1$ vector of projections of all the steering vectors on $\Phi_{l}(\theta)$. Thus, the ESPAR radiation pattern is formed by manipulating the reactive loads attached to the parasitic elements. Recalling the system model presented in the previous subsection, we assume that both the secondary base station and the PU receiver adopt a single antenna with a single basis pattern, and each of the $N$ SUs adopt an ESPAR antenna with $M$ basis patterns. As shown in \cite{13}, the linear combination of transmitted symbols leads to a linear combination of antenna responses at the receiver. Assume that the $n^{th}$ SU assigns a weight vector ${\bf w_{n}}(k) = [w_{n,1}(k), w_{n,2}(k),...,w_{n,M}(k)]$ to the ESPAR antenna basis patterns at the $k^{th}$ time slot. In this case, the received signal at the secondary base station is given by 
\[r(k) = \sqrt{P_{s}({\bf h^{i}_{sp,n^{*}}}(k), {\bf w_{n^{*}}}(k))} \left(\sum_{i=1}^{M} h^{i}_{s,n^{*}}(k) w_{n^{*},i}(k)\right) +\]
\begin{equation}
\label{3}
h_{ps}(k) x_{p}(k) + n(k),
\end{equation}
where $h^{i}_{s,n^{*}}(k)$ is the channel response between the $i^{th}$ transmit basis pattern and the base station antenna for the selected SU, and ${\bf h^{i}_{sp,n^{*}}}(k) = [h^{1}_{sp,n^{*}}(k), h^{2}_{sp,n^{*}},..., h^{M}_{sp,n^{*}}]$ is a vector of the $M$ channel responses between the $i^{th}$ transmit basis pattern of the selected SU and the PU receiver antenna. This signal model will be used for the analysis of the RAB technique presented in Section IV.

\section{Capacity Scaling of the Cognitive MAC Channel with LoS Interference}
In a conventional non-cognitive MAC channel with a transmit power constraint, the sum capacity grows with the number of users $N$ with a growth rate of $\log(\log(N))$ \cite{6}. However, for a cognitive MAC channel when only an interference constraint is imposed at the PU receiver, the SUs sum capacity grows like $\log(N)$ \cite{3}, because in addition to multiuser diversity, we can also exploit MID \cite{2}. These scaling laws are applicable for Rayleigh fading channels. In \cite{6}, it was shown that in the non-cognitive MAC system with Rician channels, the sum capacity scales like $\log\left(\left(\sqrt{\frac{1}{K+1}\log(N)}+\sqrt{\frac{K}{K+1}}\right)^{2}\right)$ instead of $\log(\log(N))$, where $K$ is the Rician $K$-factor. Thus, LoS channels hinder multiuser diversity because they experience a smaller dynamic range for the channels fluctuations. In this section, we aim at understanding the impact of LoS interference on the ability of a cognitive MAC channel to exploit MID. Because a LoS interference channel would entail a small dynamic range for the channel gain fluctuations, we expect that MID will be limited as well. The following theorem quantifies the impact of LoS interference on the capacity scaling of a cognitive MAC channel.

\begin{thm}
{\it The sum capacity of a cognitive MAC channel with LoS SU-to-PU interference grows like $\log\left(\frac{K^{2}+K}{\mathcal{W}\left(\frac{K e^{K}}{N}\right)}\right)$, where $K$ is the $K$-factor of the Rician channels, $N$ is the number of SUs, and $\mathcal{W}(.)$ is the Lambert W function.}
\end{thm}

\begin{proof} See Appendix A. \IEEEQEDhere
\end{proof}
It is clear from Theorem 1 that the capacity growth rate in LoS interference is slower than the $\log(N)$ rate experienced in Rayleigh fading interference channels. We examine the results of Theorem 1 for various values of $K$ as follows:
\begin{itemize}
\item For $K \to$ 0, and given that when $x \to 0$, $\mathcal{W}(x) \to x$ \cite{17}, we have $\mathcal{W}\left(\frac{K e^{K}}{N}\right) \to \frac{K e^{K}}{N}$, and $\frac{K(K+1)}{\frac{K e^{K}}{N}} \to N$. Thus, for $K \to$ 0 (Rayleigh fading), the SU capacity grows like $\log(N)$.
\item For moderate values \footnote{Because the term $K e^{K}$ grows exponentially with $K$, moderate values of $K$ are $K <$ 3 $\sim$ 4. For larger values of $K$, the term $K e^{K}$ becomes larger than any practical value of $N$, thus $\frac{K e^{K}}{N}$ can not to be considered small and the approximation $\mathcal{W}(x) \approx x$ would not apply.} of $K$, for $N \to \infty$, $\mathcal{W}\left(\frac{K e^{K}}{N}\right) \to \frac{K e^{K}}{N}$, and $\frac{K(K+1)}{\mathcal{W}\left(\frac{K e^{K}}{N}\right)} \to \frac{N(K+1)}{e^{K}}.$ Thus, the SU capacity scales like $\log\left(\frac{N(K+1)}{e^{K}}\right)$, which can be viewed as the same growth rate in Rayleigh channels but with an effective number of users $\frac{N(K+1)}{e^{K}}$ instead of $N$. Because $\frac{(K+1)}{e^{K}} < 1$, $\forall K > 0$, then Rician interference channels reduce MID by reducing the effective number of SUs for moderate values of $K$.
\item For large values of $K$ and practical values of $N$, the term $\frac{K e^{K}}{N} >>$ 1, thus the series expansion of $\mathcal{W}\left(\frac{K e^{K}}{N}\right)$ is given by \cite{17}
\[\mathcal{W}\left(\frac{K e^{K}}{N}\right) \approx \log\left(\frac{K e^{K}}{N}\right)-\log\left(\log\left(\frac{K e^{K}}{N}\right)\right) +O(1).\]
Thus, the SU capacity grows like $\log(\log(N))$ for large values of $K$. Thus, when the LoS component dominates ($K$ is large), the interference channels are almost deterministic and the interference constraint is effectively a transmit power constraint. Hence, the cognitive system is equivalent to a non-cognitive one, where only multiuser diversity of the secondary channels is exploited and the capacity growth is a double logarithmic function of the number of SUs.
\end{itemize}

From the above discussion, we conclude that LoS interference hinders MID. The larger the $K$-factor is, the less mutliuser interference diversity gain we attain. When $K$ is very large, the {\it interference-limited} cognitive system turns into a {\it transmit power-limited} non-cognitive one with a capacity scaling of $\log(\log(N))$. In the next section, we propose the RAB technique, and show that it can be used to exploit LoS interference and improve MID.

\section{Random Aerial Beamforming}
In order to combat the negative effects of LoS interference on the SU capacity, we propose the RAB scheme, which induces artificial fluctuations in the SU-to-PU interference channels. Because deploying cumbersome multiple antennas is not practical for mobile devices, we propose the usage of ESPAR antennas with $M$ orthogonal basis patterns. RAB is applied by the $n^{th}$ SU via adjusting the reactive load attached to the $i^{th}$ parasitic antenna at the $k^{th}$ time slot such that $w_{n,i}(k) = \sqrt{\alpha_{n,i}(k)} e^{j\theta_{n,i}(k)} x(k)$, where $x(k)$ is the transmitted symbol and $\sqrt{\alpha_{n,i}(k)} e^{j\theta_{n,i}(k)}$ is a random weight. Without loss of generality, we set $\sqrt{\alpha_{n,i}(k)}= \frac{1}{\sqrt{M}}$ and $\theta_{n,i}(k) \sim \mbox{Unif}(0,2 \pi)$. Therefore, the signal model for the selected user $n^{*}$ in (\ref{3}) can be rewritten as
\[r(k) = \sqrt{P_{s}({\bf h^{i}_{sp,n^{*}}}(k), {\bf w_{n^{*}}}(k))} \times \]
\begin{equation}
\label{5}
\underbrace{\left(\sum_{i=1}^{M} \sqrt{\alpha_{n^{*},i}(k)} e^{j\theta_{n^{*},i}(k)} h^{i}_{s,n^{*}}(k)\right)}_{h^{eq}_{s,n^{*}}} x(k) + h_{ps}(k) x_{p}(k) + n(k)
\end{equation}
where $h^{eq}_{s,n^{*}}$ is the equivalent SU-to-SU channel after applying RAB, and the transmit power $P_{s}({\bf h^{i}_{sp,n^{*}}}(k), {\bf w_{n^{*}}}(k))$ is given by
\begin{equation}
\label{7}
P_{s}({\bf h^{i}_{sp,n^{*}}}(k), {\bf w_{n^{*}}}(k)) = \frac{Q_{p}}{\left|\frac{1}{\sqrt{M}}\sum_{i=1}^{M} e^{j\theta_{n^{*},i}(k)} \, h^{i}_{sp,n^{*}}(k)\right|^{2}}.
\end{equation}
As shown in (\ref{7}), the equivalent SU-to-PU channel is the resultant of the addition of different LoS channels with random phases for each basis pattern. By varying these phases over time, the equivalent channel will experience artificial fluctuations that allow retaining the MID hindered by LoS interference. The capacity scaling for the RAB scheme is given in the following theorem.

\begin{thm}
{\it For the cognitive MAC channel with the SUs applying the RAB scheme, the capacity scales like $\log(N)$ for LoS SU-to-PU interference channels.}
\end{thm}

\begin{proof} See Appendix B. \IEEEQEDhere
\end{proof}
Theorem 2 proves that RAB can guarantee the $\log(N)$ growth rate even if the PU is subject to LoS interference. The achievability of such growth rate can be attributed to the {\it opportunistic nulling} applied by RAB through the addition of phase shifted versions of the LoS interference component. The results of Theorem 2 are valid for any number of basis patterns $M$ and for any $K$. However, the impact of $M$ and $K$ on the multiuser diversity gain is not clear. In other words, given that all values of $M$ will achieve the $\log(N)$ growth rate, which value of $M$ would result in the best MID characteristics? and what is the impact of the $K$-factor on the achievable MID gain?

An interesting result is that RAB does not only retain the $\log(N)$ growth rate in LoS channels, but it can exploit LoS interference as well. In other words, when applying RAB, the MID gain in case of LoS interference is larger than the case of Rayleigh fading interference. Moreover, the MID gain increases with the increase of the $K$-factor. Thus, RAB turns LoS interference from a source of capacity hindrance to a capacity advantage. These effects even apply when the SUs use 2 dumb basis patterns only, which entails minimal hardware complexity for the ESPAR antenna. In the following, we study the impact of $M$ and $K$ on the behavior of MID gain.

From Appendix B, we know that $h^{eq}_{s,n^{*}}$ follows a Rayleigh distribution, and the equivalent SU-to-PU channel power can be represented by (\ref{B2}). Using Euler identity, the artificial fading component can be represented as
\[\sqrt{\frac{K\overline{\gamma}_{sp}}{M(K+1)}} \sum_{i=1}^{M} \left( \cos(\theta_{n^{*},i}(k)+\phi_{i,n^{*}}) + j \sin(\theta_{n^{*},i}(k)+\phi_{i,n^{*}}) \right).\]
Noting that RAB is applied by selecting independent uniformly distributed random phases $\theta_{n^{*},i}(k) \sim \mbox{Unif}(0, 2\pi)$, then $y_{i}= \sum_{i=1}^{M} \cos(\theta_{n^{*},i}(k)+\phi_{i,n^{*}}), i = 1,2,...,M$, are independent and identically distributed. It can be easily shown that when $\theta_{i,n^{*}}(k) \sim \mbox{Unif}(0,2\pi)$, then $\Psi = (\theta_{n^{*},i}(k)+\phi_{i,n^{*}}) \sim \mbox{Unif}(\phi_{i,n^{*}}, \phi_{i,n^{*}}+2\pi)$. Because we are only interested in $\Psi \, \mbox{mod} \, 2\pi$, it can be easily shown that $\Psi \, \mbox{mod} \, 2\pi \sim \mbox{Unif}(0, 2\pi)$. Using random variable transformation, the pdf of $y_{i}$ is given by
\begin{equation}
\label{10}
f_{y_{i}}(y_{i}) = \frac{1}{\pi \sqrt{1-y_{i}^{2}}}, -1 \leq y_{i} \leq 1,
\end{equation}
with $\mathbb{E}\{y_{i}\} = 0$, and $\mathbb{E}\{(y_{i}-\mathbb{E}\{y_{i}\})^{2}\} = \frac{1}{2}$. We now study two distinct scenarios based on the number of basis patterns: $M \to \infty$, and $M$ = 2.

\subsubsection{For large number of basis patterns}
In this case, the {\it central limit theorem} applies, and $\sqrt{\frac{K\overline{\gamma}_{sp}}{M(K+1)}} \sum_{i=1}^{M}\cos(\theta_{n^{*},i}(k)+\phi_{i,n^{*}}) \sim \mathcal{N}\left(0,\frac{K\overline{\gamma}_{sp}}{2M(K+1)}\right)$. The same analysis can be applied for the term $\sqrt{\frac{K\overline{\gamma}_{sp}}{M(K+1)}} \sum_{i=1}^{M}\sin(\theta_{n^{*},i}(k)+\phi_{i,n^{*}}) \sim \mathcal{N}(0,\frac{K\overline{\gamma}_{sp}}{2M(K+1)})$, thus the artificial fading component follows a zero mean complex gaussian distribution with variance $	\frac{K\overline{\gamma}_{sp}}{M(K+1)}$. Therefore, the equivalent SU-to-PU channel $h^{eq}_{s,n^{*}} \sim \mathcal{CN}(0, \overline{\gamma}_{sp})$, which means that RAB converts the Rician channel into a Rayleigh channel. Hence, the normalizing constant in (\ref{A9}), which is proportional to the average SINR of the selected SU, can be obtained by setting $K = 0$ as follows
\begin{equation}
\label{11}
a_{N} = \frac{N\overline{\gamma}_{s}}{\overline{\gamma}_{sp}} - \frac{\overline{\gamma}_{s}}{\overline{\gamma}_{sp}} \approx  \frac{N\overline{\gamma}_{s}}{\overline{\gamma}_{sp}}.
\end{equation}
Using the derivations in Appendix A, it can be shown that the capacity grows like $\log(N)$, and the average SINR for the selected SU is approximately $\frac{N\overline{\gamma}_{s}}{\overline{\gamma}_{sp}}$. \\

\subsubsection{For $M$ = 2}
The artificial fading channel power can be expressed as
\[\frac{K\overline{\gamma}_{sp}}{2(K+1)} \left|\left( \cos(\theta_{n^{*},1}(k)+\phi_{1,n^{*}}) + \cos(\theta_{n^{*},2}(k)+\phi_{2,n^{*}}) \right. \right.\]
\[\left. \left. + j \sin(\theta_{n^{*},1}(k)+\phi_{1,n^{*}}) + j \sin(\theta_{n^{*},2}(k)+\phi_{2,n^{*}}) \right) \right|^{2},\]
which can be reduced to
\begin{equation}
\label{12}
\frac{K\overline{\gamma}_{sp}}{(K+1)} \left(1+\cos(\theta_{n^{*},1}(k)+\phi_{1,n^{*}} +\theta_{n^{*},2}(k)+\phi_{2,n^{*}}) \right).
\end{equation}
Let $\psi = \theta_{n^{*},1}(k)+\phi_{1,n^{*}} +\theta_{n^{*},2}(k)+\phi_{2,n^{*}}$. When $\theta_{n^{*},1}(k) \sim \mbox{Unif}(0,2\pi)$ and $\theta_{n^{*},2}(k) \sim \mbox{Unif}(0,2\pi)$, it can be shown that $\psi \mbox{mod} 2\pi \sim \mbox{Unif}(0,2\pi)$. Thus, the term $\cos(\theta_{n^{*},1}(k)+\phi_{1,n^{*}} +\theta_{n^{*},2}(k)+\phi_{2,n^{*}})$ follows the same pdf in (\ref{10}). The normalizing constant (average SINR for the selected user) is given by the following lemma.

\begin{lem}
{\it Applying RAB with 2 dumb basis patterns, the capacity of a cognitive MAC channel with LoS interference grows according to
\[\log \left(\sqrt{\frac{(K+1)^{2}}{2\pi K}}N\right).\]
}
\end{lem}

\begin{proof} See Appendix C. \IEEEQEDhere
\end{proof}

From Lemma 1, we conclude that when using only 2 basis patterns, the capacity follows the same growth rate of a cognitive MAC channel with Rayleigh faded interference channel, but with an effective number of SUs that is equal to $\approx \sqrt{\frac{(K+1)^{2}}{2\pi K}}N$ instead of $N$. The effective number of users increases with the increase of $K$. Thus, LoS interference can be exploited to improve MID gain and not only to restore the $\log(N)$ growth rate. An important point is that this capacity advantage is achieved with minimal hardware complexity, i.e., only 2 basis pattern, which is realized uing one parasitic antenna element and one active element.

\begin{figure}[t]
\centering
\includegraphics[width=3in]{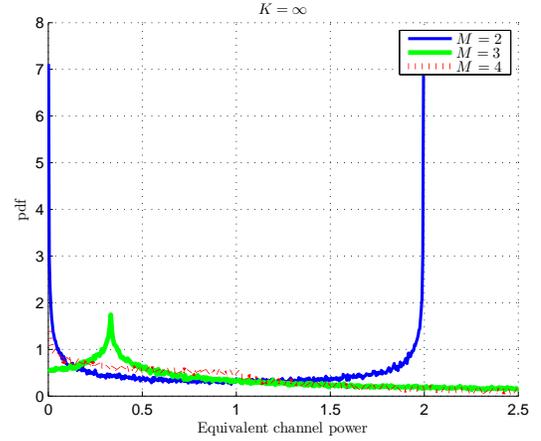}
\caption{Impact of the number of basis patterns on the pdf of the equivalent SU-to-PU channel power.}
\label{fig_sim}
\end{figure}

\begin{figure}[t]
\centering
\includegraphics[width=3in]{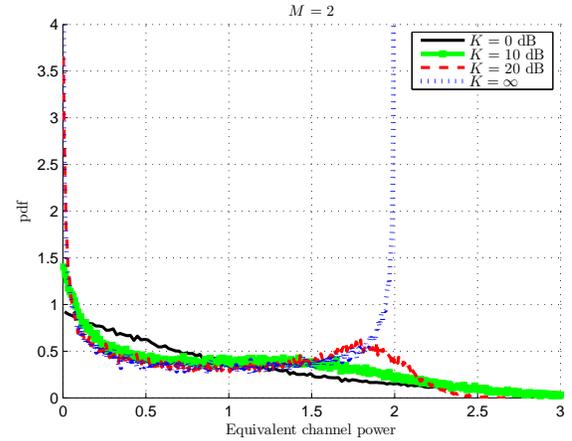}
\caption{Impact of the $K$-factor on the pdf of the equivalent SU-to-PU channel power.}
\label{fig_sim}
\end{figure}

In order to interpret the capacity advantage offered by the 2 basis patterns scenario and quantified by Lemma 1, we examine the equivalent SU-to-PU channel power in (\ref{B2}). It is apparent that the SU-to-PU channel consists of an artificial fading component and a scattered component. In Fig. 3, we set $K = \infty$ such that only the artificial fading component exists, and plot the pdf of the channel power for various values of $M$. It is clear that for $M = 2$, the pdf is concentrated around zero and the maximum value of 2. Thus, the interference channel is almost nulled half of the time, and the probability that the SU-to-PU attains an arbitrarily small value is high. As $M$ increases, the central limit theorem applies, and the pdf of the SU-to-PU channel power approaches an exponential distribution. Such distribution entails a {\it larger dynamic range}, but {\it less frequent nulling}, becaue the probability that the SU-to-PU attains an arbitrarily small value is higher in the case of $M = 2$ than in the case of $M > 2$. Therefore, using 2 basis patterns achieves desirable statistics for the interference channel.

Fig. 4 depicts the impact of the $K$-factor on the pdf of the SU-to-PU channel by plotting the pdf of the SU-to-PU channel power for $M$ = 2 and various values of $K$. Larger $K$-factors imply that the artificial fading channel dominates, and the desirable statistics created by the 2 basis patterns are more apparent in the pdf of the channel power. As the $K$-factor decreasess, the scattered component, which follows a Rayleigh distribution regardless of the RAB scheme applied, will dominate. This means that the SU-to-PU channel power will be almost exponential for small values of $K$. Therefore, larger $K$-factors imply that the interference channel is {\it SU-controlled}, and the SU can create desirable artificial channel statistics that nulls the interference more frequently. When the $K$-factor decreases, the SU-to-PU channel statistics are beyond the SU control, and the dominant scattered component will impose a quasi-exponential distribution.

\section{Numerical Results}

\begin{figure}[t]
\centering
\includegraphics[width=3.5in]{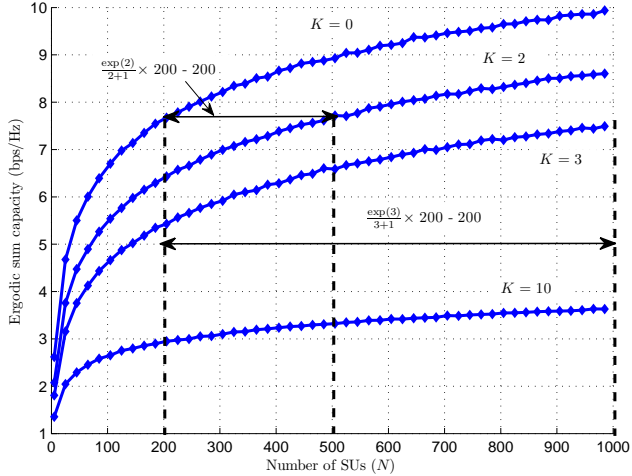}
\caption{Impact of LoS interference on the capacity of cognitive MAC.}
\label{fig_sim}
\end{figure}

\begin{figure}[t]
\centering
\includegraphics[width=3.5in]{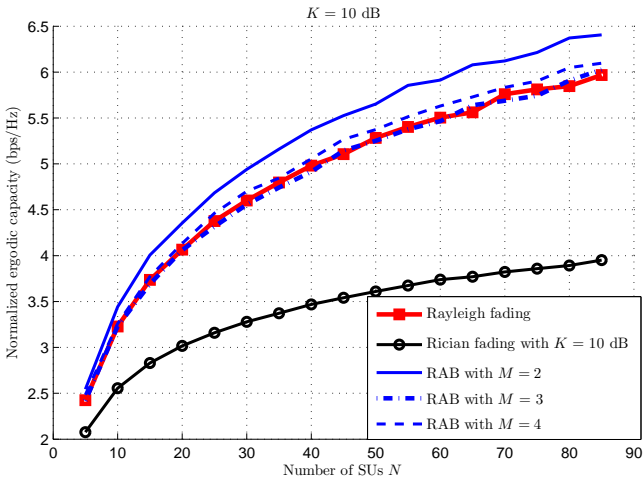}
\caption{Impact of the number of basis patterns on multiuser diversity gain.}
\label{fig_sim}
\end{figure}

\begin{figure}[t]
\centering
\includegraphics[width=3.5in]{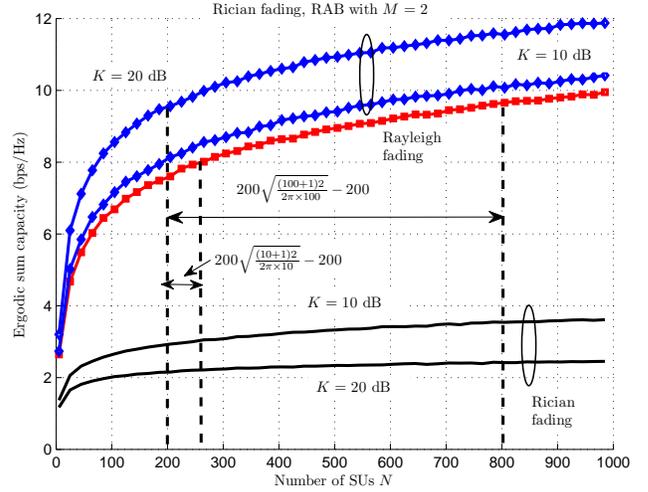}
\caption{Impact of the $K$-factor on the achieved multiuser diversity gain.}
\label{fig_sim}
\end{figure}

\begin{figure}[t]
\centering
\includegraphics[width=3.5in]{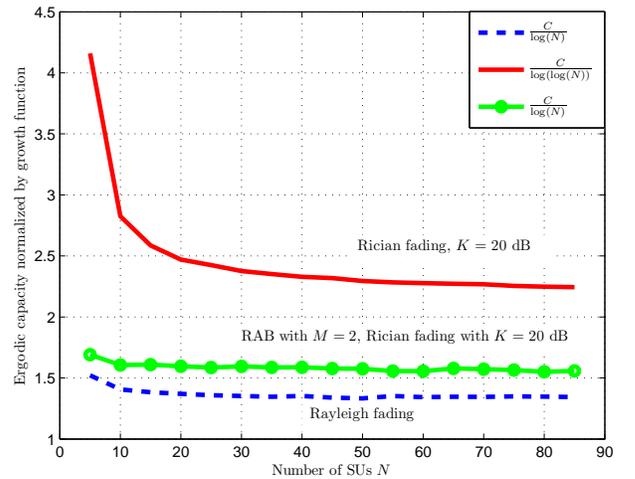}
\caption{Capacity growth rate for various schemes.}
\label{fig_sim}
\end{figure}

Monte-Carlo simulations are carried out and results are averaged over 100,000 runs. For all curves presented herein, the peak interference constraint is set to 1. In Fig. 5, we verify the results of Theorem 1. The ergodic sum capacity for the cognitive MAC network with Rician LoS interference is plotted versus the number of SUs ($N$) for $K$ = 0, 2, 3, and 10. It is clear that when the $K$-factor is non-zero, the capacity is less than the Rayleigh fading scenario, which corresponds to the case when $K$ = 0. This effect is attributed to the fact that MID gain is hindered by LoS interference. For moderate values of $K$, i.e., $K$ = 2 and 3, it is clear that the capacity follows the same logarithmic growth rate experienced by a Rayleigh channel. This is evident from the fact that the capacity curves for $K$ = 0, 2, and 3 are parallel at large values of $N$. MID hindrance is demonstrated by a reduction in the {\it effective number of users}, where the growth rate can be approximated as $\log\left(\frac{N(K+1)}{\exp(K)}\right)$. For $K$ = 2, the number of effective users is less than the actual number of users by a factor of $\frac{2+1}{\exp(2)} \approx 0.4$. Thus, the capacity achieved by 500 SUs when $K$ = 2 is equal to that achieved by 500 $\times$ 0.4 = 200 SUs for $K$ = 0. The difference between the actual and effective number of users for 200 effective SUs is depicted in Fig. 5 at $K$ = 2 and 3. For $K$ = 10, the growth rate is approximated as $\log(\log(N))$. It is clear that the capacity curve for $K$ = 10 is not parallel to the other curves as it grows in a double logarithmic fashion. In Fig. 6, the impact of the number of basis patterns used in RAB is demonstrated by plotting the ergodic capacity for $N$ users normalized to the single user capacity, which represents the multiuser diversity gain, for a cognitive MAC and a reference non-cognitive multiuser network. It is clear that using 2 dumb basis patterns, the capacity is better than that achieved when using 3 and 4 basis patterns. When the number of basis patterns increase, the equivalent interference channel converges to a Rayleigh distribution, and the performance is similar to the $K$ = 0 scenario. Fig. 7 demonstrates the advantages of applying RAB to a cognitive MAC channel with LoS interference. For $K$ = 10 and 20 dB, the ergodic SU capacity severely degrades compared to the Rayleigh fading scenario due to the double logarithmic growth and MID hindrance. The larger the $K$-factor is, the less sum capacity is achieved. Applying RAB with 2 basis patterns, the logarithmic growth rate is restored. In addition to that, it is shown that as the $K$ factor increases, the sum capacity improves, and is actually superior to the Rayleigh fading capacity. This capacity improvement is quantified by Lemma 1, where it is shown that the {\it effective number of users} increases by a factor of $\sqrt{\frac{(K+1)^{2}}{2\pi K}}$. For 200 actual SUs, the effective number of SUs is $200 \sqrt{\frac{(10+1)^{2}}{2\pi \times 10}} \approx 280$ for $K$ = 10, and $200 \sqrt{\frac{(100+1)^{2}}{2\pi \times 100}} \approx 800$ for $K$ = 100 (20 dB). In Fig. 8, the growth rates for various schemes are justified by plotting the ergodic capacity normalized by the growth function. It is obvious that the capacity of the Rayleigh fading interference channel scenario grows according to $\log(N)$, while the capacity for a LoS interference scenario with $K$ = 20 dB grows according to $\log(\log(N))$. After applying RAB, the $\log(N)$ growth rate is restored.

\section{Conclusion}           	
In this paper, it is shown that the capacity of the cognitive MAC channel with LoS interference grows according to $\log\left(\frac{K^{2}+K}{\mathcal{W}\left(\frac{K e^{K}}{N}\right)}\right)$, where $K$ is the Rician $K$-factor. This growth rate can be approximated as $\log\left(\frac{N(K+1)}{\exp(K)}\right)$, which corresponds to the same growth rate of a cognitive MAC channel with non-LoS interference but with an $\frac{N(K+1)}{\exp(K)}$ effective number of SUs. At large values of $K$, the capacity tends to grow double logarithmically with $N$. In order to improve capacity scaling in LoS interference channel, we introduced the {\it Random Aerial Beamforming} (RAB) technique. RAB induces artificial fluctuations in the interference channels by randomizing the weights of the basis patterns of a single radio ESPAR antenna. Using RAB with only 2 basis patterns, it is shown that the capacity grows like $\log \left(\sqrt{\frac{(K+1)^{2}}{2K \pi}} N\right)$. Thus, RAB can be used to exploit LoS interference and improve multiuser diversity by boosting the {\it effective number of SUs}.

\appendices

\section{Proof of Theorem 1}
\renewcommand{\theequation}{\thesection.\arabic{equation}}
We know that when the SUs use single conventional antennas, the SINR is given by $\max_{n}\frac{\gamma_{s,n} \frac{Q_{p}}{\gamma_{sp,n}}}{1+\overline{\gamma}_{p} \gamma_{ps}}$. Thus, the SU ergodic capacity is given by
\begin{equation}
\label{A1}
C = \mathbb{E}\left\{\log\left(1+ \max_{n}\frac{\gamma_{s,n} \frac{Q_{p}}{\gamma_{sp,n}}}{1+\overline{\gamma}_{p} \gamma_{ps}}\right)\right\},
\end{equation}
where $\mathbb{E}\{.\}$ is the expectation operator. Using Jensen's Inequality and the concavity of the logarithmic function, the SU capacity is bounded by
\begin{equation}
\label{A2}
C \leq \log\left(1+ \mathbb{E}\left\{ \frac{Q_{p}}{1+\overline{\gamma}_{p} \gamma_{ps}} \right\} \mathbb{E}\left\{ \max_{n} \frac{\gamma_{s,n}}{\gamma_{sp,n}} \right\}\right),
\end{equation}
which follows from the independence of the SU-to-PU, PU-to-SU, and SU-to-SU channels. The term $\mathbb{E}\left\{ \frac{Q_{p}}{1+\overline{\gamma}_{p} \gamma_{ps}} \right\}$ is a constant and independent of the SU selection, so it does not affect the capacity scaling. We define $z_{n} = \frac{\gamma_{s,n}}{\gamma_{sp,n}}$, thus we are interested in obtaining the probability density function of $z_{n}$. The SU-to-PU channel is Rician, while the secondary channel is assumed to follow a Rayleigh distribution. Thus the pdf of $\gamma_{sp,n}$ is given by \cite{14}
\[f_{\gamma_{sp,n}}(\gamma_{sp,n}) =\]
\begin{equation}
\label{A3}
\frac{(1+K)e^{-K-\frac{(1+K)}{\overline{\gamma}_{sp}} \gamma_{sp,n}}}{\overline{\gamma}_{sp}}  I_{o}\left(2 \sqrt{\frac{K(1+K)}{\overline{\gamma}_{sp}}\gamma_{sp,n}}\right)
\end{equation}
where $I_{o}(.)$ is the modified Bessel function of the first kind and that of $\gamma_{s,n}$ is exponential, i.e., $f_{\gamma_{s,n}}(\gamma_{s,n}) = \frac{1}{\overline{\gamma}_{s}} e^{\frac{-\gamma_{s,n}}{\overline{\gamma}_{s}}}$. Noting that the pdf of $\gamma_{sp,n}$ can be rewritten in terms of the Meijer-$G$ function $\MeijerG{m}{n}{p}{q}{a_1,\ldots,a_p}{b_1,\ldots,b_q}{z}$ [15, Sec. 7.8] as
\[f_{\gamma_{sp,n}}(\gamma_{sp,n}) =\]
\begin{equation}
\label{A5}
 \frac{1+K}{\overline{\gamma}_{sp}} e^{-K-\frac{(1+K)\gamma_{sp,n}}{\overline{\gamma}_{sp}}} \MeijerG{1}{0}{0}{2}{-}{0, 0}{\frac{K (1+K) \gamma_{sp,n}}{\overline{\gamma}_{sp}}}
\end{equation}
and that the pdf of $Z = \frac{X}{Y}$ is given by $f_{Z}(z) = \int_{-\infty}^{+\infty}|y| p_{x,y}(zy, y) dy$ \cite{14}, the pdf of $z_{n}$ can be obtained as shown in (A.11) using the property $z \MeijerG{m}{n}{p}{q}{a_1,\ldots,a_p}{b_1,\ldots,b_q}{z} = \MeijerG{m}{n}{p}{q}{a_1+1,\ldots,a_p+1}{b_1+1,\ldots,b_q+1}{z}$, and then evaluating the resulting integral using the standard laplace transform of a Meijer-$G$ function \cite{15}.

\begin{figure*}[!t]
\normalsize
\setcounter{mytempeqncnt}{\value{equation}}
\setcounter{equation}{12}
\begin{align}
\label{A7}
f_{z_{n}}(z_{n}) &= \frac{1+K}{\overline{\gamma}_{sp} \overline{\gamma}_{s}} \,\,\, e^{-K} \int_{0}^{\infty} \gamma_{sp,n} e^{-\left(\frac{(1+K)}{\overline{\gamma}_{sp}}+\frac{z,n}{ \overline{\gamma}_{s}}\right) \gamma_{sp,n}} \MeijerG{1}{0}{0}{2}{-}{0, 0}{\frac{K (1+K)\gamma_{sp,n}}{\overline{\gamma}_{sp}}} d \gamma_{sp,n} \\
 &= \frac{2 e^{-K}}{K \overline{\gamma}_{s}} \int_{0}^{\infty} e^{-\left(\frac{(1+K)}{\overline{\gamma}_{sp}}+\frac{z_{n}}{ \overline{\gamma}_{s}}\right) \gamma_{sp,n}} \MeijerG{1}{0}{0}{2}{-}{1, 1}{\frac{K (1+K)\gamma_{sp,n}}{\overline{\gamma}_{sp}}} d \gamma_{sp,n} \\
&= \frac{2 e^{-K}}{K \overline{\gamma}_{s}} \times \frac{1}{\frac{1+K}{\overline{\gamma}_{sp}}+\frac{z_{n}}{\overline{\gamma}_{s}}} \times \MeijerG{1}{1}{1}{2}{0}{1, 1}{\frac{K (1+K)}{\overline{\gamma}_{sp}\left(\frac{(1+K)}{\overline{\gamma}_{sp}}+\frac{z_{n}}{\overline{\gamma}_{s}}\right)}} \\
&= \frac{(1+K) \overline{\gamma}_{sp}}{\overline{\gamma}_{s}} \,\,\, e^{-K+\frac{K(1+K)}{(1+K)+\frac{\overline{\gamma}_{sp}}{\overline{\gamma}_{s}}z_{n}}} \left(\frac{(1+K)^{2}+z_{n} \frac{\overline{\gamma}_{sp}}{\overline{\gamma}_{s}}}{\left((1+K)+z_{n} \frac{\overline{\gamma}_{sp}}{\overline{\gamma}_{s}}\right)^{3}}\right).
\end{align}
\setcounter{equation}{\value{mytempeqncnt}+4}
\hrulefill
\vspace*{4pt}
\end{figure*}
The cdf of $z_{n}$ follows directly by integrating (A.16), and is given by
\begin{equation}
\label{A8}
F_{z_{n}}(z_{n}) = 1 - \frac{1+K}{z_{n}\frac{\overline{\gamma}_{sp}}{\overline{\gamma}_{s}}+K+1} \, e^{-K+\frac{K(1+K)}{\frac{\overline{\gamma}_{sp}}{\overline{\gamma}_{s}}z_{n}+K+1}}.
\end{equation}
It can be easily shown that $\lim_{z_{n} \to \infty} \frac{z_{n} f_{z_{n}}(z_{n})}{1-F_{z_{n}}(z_{n})} = \frac{\overline{\gamma}_{sp}}{\overline{\gamma}_{s}}$. Therefore, according to \cite{16}, we can find a sequence of real numbers $\{a_{N}\}_{N=1}^{\infty}$ such that $\frac{z_{N}^{*}}{a_{N}}$ converges in distribution to a {\it Frechet distributed} random variable, where $z_{N}^{*} = \max_{1 \leq n \leq N} z_{n}$, which implies that $\mathbb{E}\{z_{N}^{*}\}$ grows like $a_{N}$. That is, we have $\lim_{N \to \infty} F_{z_{N}^{*}}(z_{N}^{*}) = e^{\frac{-1}{z_{N}^{*}}}$, and the normalizing constant can be obtained as $F_{z_{n}}(a_{N}) = 1-\frac{1}{N}$ \cite{16}. Using the CDF in (\ref{A8}), we can obtain $a_{N}$ in closed-form as
\begin{equation}
\label{A9}
a_{N} = \frac{K^{2}+K}{\frac{\overline{\gamma}_{sp}}{\overline{\gamma}_{s}} \mathcal{W}\left(\frac{Ke^{K}}{N}\right)} - \frac{\overline{\gamma}_{s}(K+1)}{\overline{\gamma_{sp}}},
\end{equation}
where $\mathcal{W}(.)$ is the Lambert W function. Note that the convergence in distribution for the maximum of nonnegative random variables results in moment convergence \cite{18}. Recalling (\ref{A2}), we conclude that the SU capacity grows like $\log\left(\frac{K^{2}+K}{\mathcal{W}\left(\frac{K e^{K}}{N}\right)}\right)$, which concludes the proof.

\section{Proof of Theorem 2}
\renewcommand{\theequation}{\thesection.\arabic{equation}}
The channel between the $i^{th}$ basis pattern of the SU antenna and the PU receiver follows a Rician distribution and can be modeled as
\[h^{i}_{sp,n^{*}}(k) = \sqrt{\frac{K \overline{\gamma}_{sp}}{K+1}} e^{j \phi_{i,n^{*}}} + \sqrt{\frac{\overline{\gamma}_{sp}}{K+1}} b_{i,n^{*}}(k),\]
where $n^{*}$ is the selected SU index, $\phi_{i,n^{*}}$ is the deterministic phase of the LoS component for the $i^{th}$ basis pattern, and $b_{i,n^{*}}(k) \sim \mathcal{CN}(0, \overline{\gamma}_{sp})$ is the Rayeligh distributed scattered component.
The power of the equivalent SU-to-PU channel $\gamma^{eq}_{sp,n^{*}}(k)$ after applying RAB is given by
\begin{equation}
\label{B1}
\left|\frac{1}{\sqrt{M}}\sum_{i=1}^{M} e^{j\theta_{n^{*},i}(k)} \left(\sqrt{\frac{K \overline{\gamma}_{sp}}{K+1}} e^{j \phi_{i,n^{*}}} + \sqrt{\frac{\overline{\gamma}_{sp}}{K+1}} b_{i,n^{*}}(k)\right)\right|^{2},
\end{equation}
which reduces to
\begin{equation}
\label{B2}
\left|\underbrace{\frac{1}{\sqrt{M}}\sum_{i=1}^{M} \left(\sqrt{\frac{K \overline{\gamma}_{sp}}{K+1}} e^{j (\theta_{n^{*},i}(k)+\phi_{i,n^{*}})}\right)}_{\mbox{Artificial Fading}} + \sqrt{\frac{1}{K+1}} c_{n^{*}}(k)\right|^{2},
\end{equation}
where $c_{n^{*}}(k)$ is the equivalent scattering component after applying RAB, which was shown in \cite{6} to preserve its statistics, thus $c_{n^{*}}(k) \sim \mathcal{CN}(0, \overline{\gamma}_{sp})$. We define the following parameters $a = \frac{K}{K+1}$ and $v = \frac{1}{K+1}$, representing the power of the specular and scattered components, respectively. The minimum value of the artificial fading component is 0, when the random phases of the weights of the basis patterns add destructively. For a system with large number of SUs, and for a fixed infinitesimal $\delta > 0$, there exists almost surely a fraction $\epsilon$ of users for which the magnitude of the artificial fading channel satisfies
\[\left|\frac{1}{\sqrt{M}}\sum_{i=1}^{M} \left(\sqrt{a \overline{\gamma}_{sp}} e^{j (\theta_{n^{*},i}(k)+\phi_{i,n^{*}})}\right)\right| < \delta.\]
These $\epsilon N$ users can be thought of as experiencing Rician fading with the norm of the LoS component equal to $\delta$ instead of $a$. Note that the growth rate of the conventional MAC channel with LoS interference given in Theorem 1 can be written in terms of $a$ and $v$ as $\log\left(\frac{a}{v^{2}\mathcal{W}\left(\frac{a e^{\frac{a}{v}}}{v N}\right)}\right)$. It is highly likely that the SU with the maximum SINR would belong to the set of $\epsilon N$ having an infinitesimally small LoS component. Therefore, we can approximate the D-TDMA scheme as picking the SU with the maximum SINR among the $\epsilon N$ SUs. The SU capacity of such scheme grows at least as fast as $\log\left(\delta \left(v^{2}\mathcal{W}\left(\frac{\delta e^{\frac{\delta}{v}}}{v \epsilon N}\right)\right)^{-1}\right)$. For $\delta \to 0$, $e^{\frac{\delta}{v}} \to 1$, and $\mathcal{W}\left(\frac{\delta e^{\frac{\delta}{v}}}{v \epsilon N}\right) \to \frac{\delta e^{\frac{\delta}{v}}}{v \epsilon N}$, the growth rate tends to
\begin{equation}
\label{B5}
\log\left(\frac{\epsilon N}{v}\right).
\end{equation}
Thus, by applying RAB, the SU capacity grows logarithmically with $N$.

\section{Proof of Lemma 1}
\renewcommand{\theequation}{\thesection.\arabic{equation}}
Combining (\ref{12}) and (\ref{B2}), the equivalent SU-to-PU can be thought of as experiencing Rician fading with the power of the LoS component having a random power that changes over time. A conventional Rician channel is defined via the $a$ and $v$ parameters as explained in Appendix B, where $a$ represents the power of the LoS component and $v$ is the power of the scattered component. Therefore, the equivalent channel is Rician for a given value of $a$. Let $\tilde{a}$ be the random power of the LoS component, where $\tilde{a} = \left(1+\cos(\theta_{n,1}(k)+\phi_{1,n} +\theta_{n,2}(k)+\phi_{2,n}) \right)$ as explained in Section II. We know that $\cos(\theta_{n,1}(k)+\phi_{1,n} +\theta_{n,2}(k)+\phi_{2,n})$ follows the pdf in (\ref{10}). Using random variable transformation, it can be shown that
\begin{equation}
\label{C1}
f_{\tilde{a}}(\tilde{a}) = \frac{K+1}{\pi K \sqrt{1-\left(1-\frac{\tilde{a}(K+1)}{K}\right)^{2}}}, 0 \leq \tilde{a} \leq 2\frac{K}{K+1}.
\end{equation}
Let $\gamma^{eq}_{sp,n}$ be the equivalent channel power after applying RAB for the $n^{th}$ SU. Given that the equivalent SU-to-SU channel is Rayleigh distributed after applying, and defining $z^{eq}_{n} = \frac{\gamma^{eq}_{s,n}}{\gamma^{eq}_{sp,n}}$, it follows from (\ref{A8}) that the cdf of $z^{eq}_{n}$ for a given $\tilde{a}$
\begin{equation}
\label{C2}
F_{z^{eq}_{n}}(z^{eq}_{n}|\tilde{a}) = 1 - \frac{1/v}{z^{eq}_{n}\frac{\overline{\gamma}_{sp}}{\overline{\gamma}_{s}}+\frac{1}{v}} \, \exp\left(-\frac{\tilde{a}}{v}+\frac{\tilde{a}}{\frac{\overline{\gamma}_{sp}}{\overline{\gamma}_{s}}z^{eq}_{n}v^{2}+v}\right).
\end{equation}
The conditional cdf in (\ref{C2}) is a reformulation of (\ref{A8}) in terms of $a$ and $v$ instead of $K$. Let $p = \frac{1}{v}-\frac{1}{\frac{\overline{\gamma}_{sp}}{\overline{\gamma}_{s}}z^{eq}_{n}v^{2}+v}$, the cdf of $z^{eq}_{n}$ is obtained by averaging $F_{z^{eq}_{n}}(z^{eq}_{n}|\tilde{a})$ over the pdf of $\tilde{a}$ as follows
\[F_{z^{eq}_{n}}(z^{eq}_{n}) =\]
\[1 - \frac{1/v}{z^{eq}_{n}\frac{\overline{\gamma}_{sp}}{\overline{\gamma}_{s}}+\frac{1}{v}} \int_{0}^{2\frac{K}{K+1}} \frac{(K+1) \exp\left(-p \tilde{a}\right)}{\pi K \sqrt{1-\left(1-\frac{\tilde{a}(K+1)}{K}\right)^{2}}} d \tilde{a}, \]
applying the substitution $y = 1-\frac{\tilde{a}(K+1)}{K}$, the integral reduces to
\[1 - \frac{1/v}{z^{eq}_{n}\frac{\overline{\gamma}_{sp}}{\overline{\gamma}_{s}}+\frac{1}{v}} \exp\left(-p \frac{K}{K+1}\right) \int_{-1}^{1} \frac{ \exp\left(p \frac{K}{K+1}y\right)}{\pi \sqrt{1-y^{2}}} dy, \]
by applying another substitution as $y = \cos(\theta)$, the integral simplifies as
\[1 - \frac{1/v}{z^{eq}_{n}\frac{\overline{\gamma}_{sp}}{\overline{\gamma}_{s}}+\frac{1}{v}} \times \frac{\exp\left(-p \frac{K}{K+1}\right)}{\pi} \int_{0}^{\pi} \exp\left(p \frac{K}{K+1}\cos(\theta)\right) d \theta, \]
which is derived in closed-form in [15, eq. (3.339)] as
\[F_{z^{eq}_{n}}(z^{eq}_{n}) =\]
\begin{equation}
\label{C3}
1 - \frac{1/v}{z^{eq}_{n}\frac{\overline{\gamma}_{sp}}{\overline{\gamma}_{s}}+\frac{1}{v}} \exp\left(-p \frac{K}{K+1}\right) I_{o}\left(p \frac{K}{K+1}\right).
\end{equation}
The secondary base station selects the SU with the maximum $z^{eq}_{n}$. Let $z^{eq}_{n^{*}} = \max_{1 \leq n \leq N} z^{eq}_{n}$, the statistics of $z^{eq}_{n^{*}}$ depend on the tail pdf and cdf of $z^{eq}_{n}$ \cite{6}. It can be shown that for $z^{eq}_{n} \to \infty$, $p \to \frac{1}{v}$, and $\exp\left(-p \frac{K}{K+1}\right) I_{o}\left(p \frac{K}{K+1}\right) \approx \frac{1}{\sqrt{2 \pi K}}$. Thus, the tails of the cdf and pdf are given by
\begin{equation}
\label{C5}
F_{z^{eq}_{n}}(z^{eq}_{n}) \sim 1 - \frac{1/v}{z^{eq}_{n}\frac{\overline{\gamma}_{sp}}{\overline{\gamma}_{s}}+\frac{1}{v}} \times \frac{1}{\sqrt{\pi K}},
\end{equation}
and
\begin{equation}
\label{C5}
f_{z^{eq}_{n}}(z^{eq}_{n}) \sim \frac{\frac{\overline{\gamma}_{sp}}{v\overline{\gamma}_{s}}}{\left(z^{eq}_{n}\frac{\overline{\gamma}_{sp}}{\overline{\gamma}_{s}}+\frac{1}{v}\right)^{2}} \times \frac{1}{\sqrt{\pi K}}.
\end{equation}
The asymptotic expression in (\ref{C5}) is verified by Fig. 9. It can be easily shown that $\lim_{z^{eq}_{n} \to \infty} \frac{z^{eq}_{n} f_{z^{eq}_{n}}(z^{eq}_{n})}{1-F_{z^{eq}_{n}}(z^{eq}_{n})} = \frac{\overline{\gamma}_{sp}}{\overline{\gamma}_{s}}$. Therefore, we can find a sequence of real numbers $\{a_{N}\}_{N=1}^{\infty}$ such that $\frac{z^{eq}_{n^{*}}}{a_{N}}$ converges in distribution to a {\it Frechet distributed} random variable. Following the analysis in Appendix A, and using the cdf in (\ref{C5}), we can obtain $a_{N}$ in closed-form as
\[a_{N} \approx \sqrt{\frac{(K+1)^{2}}{2\pi K}}\frac{N\overline{\gamma}_{s}}{\overline{\gamma}_{sp}}.\]

\begin{figure}[t]
\centering
\includegraphics[width=3in]{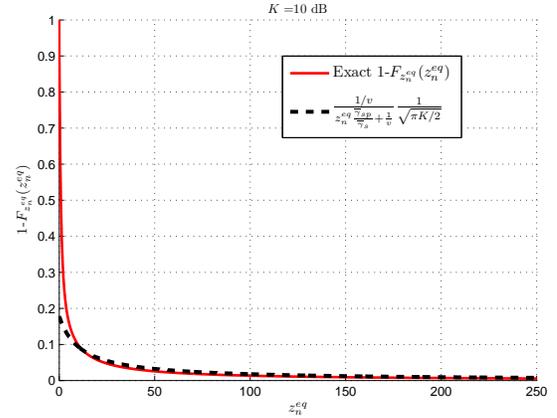}
\caption{Approximation for the cdf of the equivalent SU-to-PU channel power.}
\label{fig_sim}
\end{figure}

\end{document}